\documentclass[letterpaper, 10 pt, conference]{ieeeconf}  
\IEEEoverridecommandlockouts     
\newtheorem{mydef}{Definition}
\newtheorem{mythm}{Theorem}
\newtheorem{myprob}{Problem}

\newtheorem{remark}{Remark}
\newtheorem{assumption}{Assumption}

\usepackage{amsfonts}
\DeclareMathSymbol{\shortminus}{\mathbin}{AMSa}{"39}

\usepackage{flushend}
\usepackage{multirow}
\usepackage{makecell}
\usepackage{booktabs}
\usepackage{array}

\usepackage{graphicx}
\graphicspath{{./Pics/}}
\DeclareGraphicsExtensions{.pdf}
\usepackage{float}
\usepackage{amsmath}
\usepackage{amsfonts}
\usepackage{epsfig}
\usepackage{graphicx}
\usepackage{arydshln}
\usepackage{algpseudocode}
\usepackage{verbatim}
\usepackage{subfigure}
\usepackage{enumerate}
\usepackage{rotating}
\usepackage{color}
\usepackage{caption} 
\usepackage{mathrsfs}
\usepackage{amssymb}
\usepackage{mathtools}
\usepackage[linesnumbered,ruled,vlined]{algorithm2e}
\usepackage{cite}
\usepackage{float}
\usepackage{soul}
\usepackage{tikz}
\usepackage{hyperref}
\hypersetup{hidelinks} 
\usepackage{soul}
\usepackage{booktabs}
\usepackage{multirow}
\usepackage{makecell}
\usepackage{array}
\usetikzlibrary{arrows.meta}
\usetikzlibrary{patterns}
\usepackage{tcolorbox}
\usepackage{fancyhdr}

\floatname{algorithm}{Procedure}

\title{\LARGE \bf
Data-Driven Safe Controller Synthesis for Deterministic Systems: \\ A Posteriori Method With Validation Tests}

\author{Yu Chen, Chao Shang, Xiaolin Huang and Xiang Yin% 
\thanks{This work was supported by  the National Natural Science Foundation of China (62061136004, 62173226, 61833012).}
	\thanks{Yu Chen, Xiaolin Huang and Xiang Yin are with Department of Automation and Key Laboratory of System Control and Information Processing, Shanghai Jiao Tong University, Shanghai 200240, China.
	{\tt\small \{yuchen26,  yinxiang,xiaolinhuang\}@sjtu.edu.cn}. 
	Chao Shang is with the Department of Automation,
Beijing National Research Center for Information Science and Technology, Tsinghua University, Beijing 100084, China.
	{\tt\small c-shang@tsinghua. }
	} 
}

\begin{document}
	
\maketitle
\thispagestyle{empty}
\pagestyle{empty}
\setlength{\abovecaptionskip}{0pt}
\setlength{\belowcaptionskip}{3pt}
\setlength{\textfloatsep}{0pt}

\begin{abstract}
In this work, we investigate the data-driven safe control synthesis problem for unknown dynamic systems. 
We first formulate the safety synthesis problem as a robust convex program (RCP) based on notion of control barrier function. To resolve the issue of unknown system dynamic, we follow the existing approach by converting the RCP to a scenario convex program (SCP) by randomly collecting finite samples of system trajectory. However, to improve the sample efficiency to achieve a desired confidence bound, we provide a new posteriori method with validation tests.  Specifically, after collecting a set of  data for the SCP, we further collect another set of independent \emph{validate data} as  posterior information to test the obtained solution.  We derive a new overall confidence bound for the safety of the controller that  connects the original sample data, the support constraints, and the validation data.  The efficiency of the proposed approach is illustrated by a case study of room temperature control.  We show that, compared with existing methods, the proposed approach can significantly reduce the required number of sample data to achieve a desired confidence bound.  
\end{abstract}

\section{Introduction}
With the increasing complexity of engineering cyber-physical systems, ensuring safety has become a top priority in their design. This is particularly important as the consequences of failures or errors in these systems can be severe, ranging from property damage to loss of life. In order to ensure that these systems operate safely and correctly, engineers and developers often turn to \emph{formal methods}. These methods provide a rigorous framework for analyzing and verifying system behavior, and can provide provable guarantees of correctness and safety \cite{tabuada2009verification,belta2017formal}.

In the field of formal synthesis of safe controllers, there has been a significant amount of research conducted in recent years, resulting in the development of various approaches. These approaches can broadly be categorized as either abstraction-based or abstraction-free. Abstraction-based methods involve constructing a finite abstraction of the original system, typically achieved by discretizing the state space \cite{5342460,zamani2014symbolic,lavaei2022automated,liu2022secure}. Symbolic algorithms can then be applied to this abstraction to synthesize a controller, which can subsequently be refined to the original system. However, a significant drawback of this approach is the curse of dimensionality, which limits its suitability for large-scale systems. On the other hand, abstraction-free approaches for safe control synthesis are becoming increasingly popular, with one widely-used method being control barrier functions (CBF) \cite{ames2017control,wang2017safety,jagtap2020formal,nejati2022compositional,xiao2022adaptive,yang2022differentiable,xiao2022high}. Unlike abstraction-based techniques, CBF can directly synthesize a controller to enforce safety   without the need to discretize the state spaces. This approach offers advantages in terms of scalability, making it more feasible for 
high-dimensional systems.

The aforementioned  techniques for safe control synthesis rely on having knowledge of the system model, which can be costly or even impossible for complex systems. To address this issue, recent research has advocated for the use of data-driven approaches. For instance, several techniques have been developed to construct formal abstractions directly from data with confidence guarantees, as described in \cite{makdesi2021data,lavaei2022constructing,peruffo2022data}. These approaches enable the construction of a finite abstraction of the system directly from data, without requiring a priori knowledge of the system model.
Additionally, there are works that combine control barrier functions with collected data to synthesize controllers when the system model is partially or fully unknown; see, e.g.,  \cite{han2015sublinear,robey2020learning,jagtap2020control,lindemann2021learning}. These approaches offer promising avenues for safe control synthesis in scenarios where accurate models may be difficult or impossible to obtain.

Recently, there has been a surge of interest in data-driven verification and synthesis for safety, driven in part by the development of the theory \emph{scenario convex programming} \cite{calafiore2006scenario,esfahani2015performance}. This approach provides a sound method for safety verification or synthesis by connecting the number of sample data to the confidence bound. For instance, for deterministic systems, the safety verification problem has been addressed in \cite{nejati2023formal} for both discrete and continuous-time cases. 
%while the synthesis problem has also been tackled in \cite{nejati2022data} 
Additionally, in \cite{salamati2021adhs}, safety verification for stochastic systems has been investigated, and the results have been extended to the synthesis problem in \cite{salamati2021data}. 
Furthermore,  the wait-and-judge approach    \cite{salamati2022data} and  the repetitive approach   \cite{salamati2022Repetitive} have also been used to improve the sample efficiency of the safety verification problem.

% An optimization problem, called scenario convex program is proposed in~\cite{calafiore2006scenario} as a finite constraints relaxation of robust convex program which has infinitely many constraints. \cite{calafiore2006scenario} prove the relation between solution of scenario program and that of robust program by providing a bound on violation probability of constraint. \cite{esfahani2015performance} further related the optimal value of scenario program to that of robust program. There are some investigation on applying this framework such as invariance verification~\cite{wang2020scenario}, safety verification and synthesis~\cite{salamati2021data,nejati2023formal}. Our framework is similar to safety synthesis in~\cite{salamati2021data}. However, the required sample number of proposed method in~\cite{salamati2021data} is exponential with respect to dimension of system while we make use of posterior information to reduce sample complexity.

In this work, we focus on studying the data-driven control synthesis of unknown discrete-time deterministic systems for safety specifications. Our method also builds upon the tools of control barrier functions and scenario theory.  Specifically, we follow the approach in \cite{salamati2021data} by converting the safety control problem into a robust convex program (RCP) that searches for a control barrier function, ultimately solved by a scenario convex program. However, motivated by the recent results in \cite{shang2021tac}, we introduce a new mechanism called the \emph{validation test} for the control synthesis problem.
Specifically, our approach requires to collect two different data sets:
\begin{itemize}
\item We first collect   $N$ data to formulate the scenario convex program in order to obtain a solution;
\item Then, we collect   $N_0$ independent  \emph{validation data} as \emph{posterior information} to test the obtained solution such that the confidence bound can be further  improved.
\end{itemize}
In contrast to    \cite{salamati2022data} and   \cite{salamati2022Repetitive} for the verification problem, where the information of support constraints number and the information of violation frequency in validation data are used independently, 
here we not only consider the \emph{synthesis} problem, but also use these two posterior information \emph{jointly}. 
Therefore, our main result is an overall performance bound that connects all three information: the original sample data, support constraints, and validation data, in a uniform manner. We show that, compared with existing methods, the proposed approach can significantly reduce the required number of sample data to achieve a desired confidence bound.
%To our knowledge, the idea of validation tests has never been applied to the data-driven formal control synthesis problem.  

% In this work, we  study the data-driven control synthesis of   unknown discrete-time deterministic systems for safety specifications. 
% Our  method is also build upon the tools of  control barrier functions and the scenario theory. 
% In particular, we follow the approach in \cite{salamati2021data} by convert the safety control problem as the robust convex program (RCP) that search for  control barrier function, which is ultimately solve by a scenario convex program. 
% However, motivated by the recent result in \cite{shang2021tac},  we introduce a new mechanism called \emph{validation test} for the control synthesis problem. 
% Specifically, our approach requires to collect two different sets of data: 
% \begin{itemize}
%     \item 
%     We first collect a set of $N$ data to formulate the SCP in order to obtain a solution; 
%     \item 
%     Then we collect a set of $N_0$ independent data to \emph{valid} the solution as \emph{posterior information} in order to improve the confidence bound. 
% \end{itemize}
% We also use the number of support constraints as an additional information to improve the sample efficiency. 
% Our main result is to obtain an overall performance bound that connect the original sample data, support constraints and validation data. We show that, compared with exiting methods,  the proposed approach can significantly reduce the required number of sample data to achieve a desired confidence bound. 

The rest of the paper is organized as follows. In Section II, we provide the necessary preliminaries and formulate the problem. In Section III, we review the existing results on solving data-driven safe control synthesis using scenario convex programs. Our main theoretical contributions are presented in Section IV, where we describe the overall synthesis procedure and derive a new performance bound using posterior information. We demonstrate the sample efficiency of the proposed method in Section V through a room temperature control example. Finally, we conclude the paper in Section VI.

\section{Preliminary and Problem Statement}\label{sec:pre}
\subsection{Notations}
We denote by $\mathbb{R}$, $\mathbb{R}_{0}^{+}$ $\mathbb{N} \coloneqq \{ 1,2,3,\dots \}$ and $\mathbb{N}_0 \coloneqq \{ 0,1,2,\dots \}$  the set of real numbers, non-negative real numbers, positive integers and non-negative integers, respectively. The indicator function is denoted by $\mathbf{1}_{\mathcal{A}}(X) : X \rightarrow \{0,1 \}$ where $\mathbf{1}_{\mathcal{A}}(x)=1$ if and only if $x \in \mathcal{A}$. Given $N$ vectors $x_i\in \mathbb{R}^{n_i}$, $n_i \in \mathbb{N}$ and $i \in \{1,\dots, N \}$, we denote by $[x_1;\dots ; x_n]$ and $[x_1,\dots,x_n]$ the corresponding column and row vectors, respectively.
%Given a sequence $x=x_1x_2\dots x_n$, $\lvert x \rvert$ is denoted by length of $x$.
We denote by $\Vert x \Vert$ and $\Vert x \Vert_{\infty}$ the Euclidean norm and infinity norm of $x \in \mathbb{R}^n$, respectively. The induced norm of matrix $A \in \mathbb{R}^{m \times n}$ is defined by $\Vert A \Vert = \sup_{\Vert x \Vert = 1}\Vert Ax \Vert$.

We consider a probability space with the tuple $(\Delta,\mathcal{F},\mathbb{P})$ where $\Delta$ is the sample space, $\mathcal{F}$ is a $\sigma$-algebra on $\Delta$ and $\mathbb{P}$ is a probability measure defined over $\mathcal{F}$. Given $N \in \mathbb{N}$, $\mathbb{N}_0 \ni m \leq n$, and $t \in (0,1)$, the Beta cumulative probability function is defined as $B_{N}(t;m) \coloneqq \sum_{i=0}^{m} \binom{N}{i}t^{i}(1-t)^{N-i}$.

% Given a random variable $Z$, we denote by $\mathbb{E}(Z)$ the expectation of $Z$. The variance of $Z$ is denoted by $\textsf{var}(Z) \coloneqq \mathbb{E}(Z^2)- (\mathbb{E}(Z))^2$. 

\subsection{System Model}
We consider a discrete-time dynamical system  (dt-DS)
\[
\mathbf{S}=(X,U,f), 
\]
where
$X \subseteq \mathbb{R}^{n}$ is a Borel space  representing state space of system, 
$U \subseteq \mathbb{R}^{m}$ is a set of control inputs 
and  $f: X \times U \rightarrow X$ is an \emph{unknown} function describing the dynamic of the system. 
% Given an initial state $x(0) \in X$ and a sequence of control inputs $u(0),u(1),\dots\in U$, the state evolution of the system is given by
% \begin{equation}
% x(t+1)=f(x(t),u(t)),   t=0,1,\dots.
% \end{equation}
A (static state-feedback) controller is a mapping 
$C: X\to U$ 
that determines the control input based on the current state.  Given a controller $C$ and initial state $x(0) \in X$, the trajectory of the system is defined by 
\[
\zeta(x(0))=x(0)x(1)\dots x(n) \dots,
\]
where $x(t+1)=f(x(t),C(x(t)))$ for all $t=0,1,\dots$. 
For any $T \in \mathbb{N}$, 
we denote by $\zeta_{T}(x(0))=x(0)x(1)\cdots x(T)$ the finite prefix of trajectory $\zeta(x(0))$ of length $T+1$.  
We denote by $\mathbf{S}_{C}$ the closed-loop system under control. 
% XY: L is not a good notation because it is used of Lipshiz
We assume that the control input set is described as a polytope, i.e., 
\begin{equation} \label{eq:input}
    U=\{ u \in \mathbb{R}^{m} \mid  A u \leq B \}
\end{equation}
 where $A \in \mathbb{R}^{r \times m}$,  $B \in \mathbb{R}^{r}$ and $r\in \mathbb{N}$.

Although the dynamic function $f$ is   unknown,  we assume that we can \emph{simulate} the system by selecting  initializing the system at state $x\in X$, applying input $u\in U$ and observing the next state state $x'\in X$ of the system. 
Such a  tuple $(x,u,x')$ is referred to as a data. 
Suppose that we assign a distribution $\mathbb{P}$, where $\Delta=X \times U$, to sample $N$ i.i.d.\ pair $(x_i,u_i)$. Then the collected dataset $\mathcal{D} $ is 
 \begin{equation} \label{eq:dataset}
   \mathcal{D}:=  \{ (x_i,u_i,f(x_i,u_i)) \mid i=1,\dots,N  \}\subseteq X \times U \times X. 
 \end{equation}

% In this work, we are interested in synthesizing a controller $C$ such that controlled system $\mathbf{S}_{C}$ satisfies safety specification. Formally, a controller is given by $C=(C_0,C_1,\dots,C_n,\dots)$ with $C_k: H_k \rightarrow U$, where $H_n$ is a set of all $n$-histories $h_n=(x(0),u(0),x(1),u(1),\dots,x(n-1),u(n-1))$. A controller $C$ is called stationary if $C_0=C_i$ for all $i \in \mathbb{N}_0$. The stationary controller only depends on current state at every time instant and does not change over time. We only consider stationary controller in this work. Assume that the controller is stationary in the rest of paper.

%Given stationary controller $C$ and initial state $x_0$, the probability of trajectory $x=x_0x_1\dots x_n$ is defined by $\textsf{Pr}_{C}^{x_0}(x)=\prod_{i=0}^{n-1} \textsf{Pr}_{C}(x_{i+1}|x_i)$, where
%\[
%\textsf{Pr}_{C}(x_{i+1}|x_i)=\mathbb{P}_{w}(\{ \delta \in \Delta \mid x_{i+1} = f(x_i,C(x_{i}),v), w(i)(\delta)=v \})
%\]

%\subsection{Safety Specification}
%For ds-SCS $\mathbf{S}=(X,V,U,w,f)$ and property $\varphi=(X_0,X_u,T)$ where $X_0$, $X_u \subseteq X$ and $T \in \mathbb{N}$, we define the safety trajectory set initial from $x_0$ by
%\[
%S_{x_0}=\{ x_0x_1\dots x_{T-1} \in X^* | \forall i \in \{ 0,1,\dots , T-1 \}, x_i \not \in X_u  \}.
%\]

%We say that $\mathbf{S}$ satisfies $\varphi$ under controller $C$ within horizon $T$, denoted by $\mathbf{S}_{C} \models_{T} \varphi$, if $\forall x_0 \in X_0$, $\mathbb{P}_{C}^{x_0}(S_{x_0})=1$. Without loss of generality, we assume that $X_0 \cap X_u = \emptyset$.

\subsection{Problem Statement}
Given a dt-DS $\mathbf{S}=(X,U,f)$ and a $3$-tuple \emph{property}  
\[
\varphi=(X_0,X_u,T), 
\]
where 
$X_0  \subseteq X$ denotes the initial region, 
$X_u  \subseteq X$ denotes the unsafe region,  
and $T$ denotes the \emph{horizon} of the property. 
We assume that $X_0\cap X_u=\emptyset$. 
We say a trajectory is \emph{safe} if it does not contain an unsafe state in $X_u$, 
and we denote the set of safety trajectories w.r.t.\ $\varphi$ by
\[
\Xi_{\varphi}=\{ x_0\dots x_{T} \in X^T \mid \forall i \in 0,\dots,T\text{ s.t }x_i \not \in X_u  \}.
\]

Given controller $C$, we say that the closed-loop system $\mathbf{S}_{C}$ satisfies property $\varphi$, denoted by $\mathbf{S}_{C} \models \varphi$  if 
\[ 
 \forall x(0) \in X_0: \zeta_{T}(x(0)) \in \Xi_{\varphi}.
 \] 
The problem that we  solve in this work is stated as follows.
\begin{myprob} \label{problem}
Consider an unknown dt-DS $\mathbf{S}=(X,U,f)$ and a safety  property $\varphi=(X_0,X_u,T)$. Using data in the form of~(\ref{eq:dataset}) to find a controller $C : X \rightarrow U$ such that $\mathbf{S}_{C} \models \varphi$ with a confidence of $(1-\beta) \in [0,1]$, i.e.,
\[
\mathbb{P}^{N}(\mathbf{S}_{C} \models \varphi) \geq 1-\beta, 
\]
where $\mathbb{P}^{N}$ is the $N$-cartesian product of distribution $\mathbb{P}$.
\end{myprob}

\section{Scenario Approach using  Barrier Certificates}
The problem described in Problem 1 has already been addressed in the literature by \cite{salamati2021data}. 
The basic idea is to use control barrier functions (CBF) as a sufficient condition for ensuring safety properties, and then to solve a convex program to identify candidate CBFs through a scenario approach. We will briefly describe the existing method since our new approach builds upon it.
 
\begin{mydef}[control barrier functions] \label{def:CBF}
Given a dt-DS $\mathbf{S}=(X,U,f)$  and  property $\varphi=(X_0,X_u,T)$, a function $\mathcal{B} : X \rightarrow \mathbb{R}$ is said to be a  \emph{control barrier function} (CBF) for $\mathbf{S}$ and $\varphi$ if there exist constants $\lambda,\gamma \in \mathbb{R}$, $c \geq 0$, and functions $F_{\imath}(x) : X \rightarrow \mathbb{R},\imath=1,\dots,m $ with $[F_1(x);\dots;F_m(x)] \in U$ such that
\begin{align}
   &\mathcal{B} (x) < \gamma, \qquad   \qquad  \forall x 
   \in X_0, \label{barrier:condition 1} \\  
   &   \mathcal{B} (x) \geq \lambda,  \qquad \qquad  \forall x 
   \in X_u, \label{barrier:condition 2} \ \\
&   \mathcal{B} (f(x,u))+\sum_{\imath=1}^{m}(u_\imath-F_\imath(x)) \leq \mathcal{B} (x)+c,  \label{barrier:condition 3} \\
 & \qquad \qquad \qquad \qquad \   \forall x \in X, \forall (u_1,\dots,u_m) \in U,     \nonumber \\
 &    \lambda-\gamma \geq c T , 
\end{align}
\end{mydef}
\medskip

As shown in \cite{salamati2021data}, for any dt-DS, if we can find a CBF and its associated parameters, 
then controller $C:X \to U$ defined by 
\begin{equation} \label{eq:controller}
    C(x)=[F_1(x);\dots;F_m(x)],\forall x\in X.
\end{equation}
ensures the satisfication of $\mathbf{S}_{C} \models \varphi$. 

To identify a suitable CBF, a commonly adopted approach is to search among candidate polynomial functions.
Specifically, a polynomial CBF $ \mathcal{B} (q,x)$ 
with degree $k \in \mathbb{N}$ is of form 
\begin{equation} \label{eq:CBFcoe}
    \mathcal{B} (q,x)=\sum_{a_{1}=0}^{k} \dotsc \sum_{a_{n}=0}^{k} q_{a_1,\dots,a_n}(x_1^{a_1}x_2^{a_2}\dotsc x_n^{a_n}),
\end{equation}
where $q$ is is the vector for all coefficients 
and for $\sum_{i=1}^{n}a_{i} > k$, we have  $q_{a_1,\dots,a_n}=0$.
Similarly,  for each $\imath\in \{1,\dots,m\}$, a polynomial function $F_{\imath}(p_\imath,x)$ with degree $k_\imath$ is of form
\begin{equation} \label{eq:inputcoe}
    F_{\imath}(p_{\imath},x)=\sum_{a_1=0}^{k_{\imath}} \dotsc \sum_{a_m=0}^{k_{\imath}} p^{\imath}_{a_1,\dots,a_m}(x_1^{a_1}x_2^{a_2}\dotsc x_m^{a_m}),
\end{equation}
where $p_\imath$ is the vector for all coefficients 
and for $\sum_{i=m}^{n}a_{i} > k_\imath$, we have  $p^{\imath}_{a_1,\dots,a_m}=0$.  
We define $p=(p_1,\dots,p_m)$ as the overall coefficient vector. 

Then by restricting to candidate polynomial functions, one can synthesize a CBF-based controller by 
solving the following Robust Convex Program (RCP):  
\begin{align}
\text{RCP: } \left\{
\begin{aligned}
& \min_{d}\ \  K  \\
  \text{s.t.} \ \  &   \!\!\!\!  \max_{z\in \{1,2,3,4 \}}g_z(x,u,d) \leq 0,  \forall x \in x,  u \in U, \\ 
& d=(K,\lambda,\gamma,c,q,p), \\
& K \in \mathbb{R}, c\geq 0, \lambda - \gamma \geq cT
\end{aligned}
\right.
\label{eq:RCPconstraints}
\end{align} 
where
\begin{align}
     &  g_1(x,d)= (\mathcal{B} (q,x)-\gamma)\mathbf{1}_{X_0}(x), \\
     & g_2(x,d)= (-\mathcal{B} (q,x)+\lambda)\mathbf{1}_{X_u}(x), \\
    %& g_4(x,d)= \frac{1+cT}{\rho} -\lambda- K, \\
    %& g_4(d) = cT - \lambda + 1 - K \\
    & g_3(x,u,d)= \mathcal{B} (q,f(x,u))-\mathcal{B} (q,x) \\
    & \qquad \qquad \qquad + \sum_{\imath=1}^{m}(u_{\imath}-F_{\imath}(p^\imath,x)) -c-K, \label{eq:unknownexpect} \\
    & g_{4}(x,d) = \max_{i \in \{ 1,\dots,r \} } ( A[F_{1}(p_{1},x);\dots;F_{m}(p_{m},x)]-B )_i.
\end{align}

Intuitively,  if the optimal value for the above RCP, denoted by $K^*$, satisfies $K^* < 0$, then we know that $\mathcal{B}(q,x)$ is a valid CBF with associated control functions $F(p,x)=[F(p_1,x);\dots; F(p_m,x)]$. Specifically,  $g_1(\cdot)$-$g_3(\cdot)$ ensure that $\mathcal{B}(q,x)$ satisfies definition of CBF
and $g_4(\cdot)$ enforces that the selected control inputs are within the polytope  defined in (\ref{eq:input}).  
For technical purpose, we further assume that all constraints are Lipschitz continuous with respect to $x$ and $u$, 
and we denote by $L>0$ the Lipschitz constant for all $g_z,z=1,2,3,4$. 

However, the above RCP-based approach can only be used when the dynamic function $f$ is known. 
When $f$ is unknown, one can make use of the  collected dataset $\mathcal{D}$ in (\ref{eq:dataset}) to solve the RCP using the \emph{scenario approach}.  
Specifically, one needs to replace constraint $g_3(x,u,d)$ that should hold \emph{for all} $x\in X$ and $u\in U$ 
by a set of $N$ constraints based on the sampled data. This leads to the following Scenario Convex Program (SCP) 
\begin{align}
\text{SCP}_{N}\text{: } \left\{
\begin{aligned}
& \min_{d}\ \  K  \\
  \text{s.t.} \ \  &   \!\!\!\!  \max_{z\in \{1,2, 4 \}}\{ g_z(x,d),g_3(x_i,u_i,d)\} \leq 0, \\
                   &\quad \quad \quad \quad \quad  \forall x \in x, \forall i\in \{1,\dots, N\} \\ 
& d=(K,\lambda,\gamma,c,q,p), \\
& K \in \mathbb{R}, c\geq 0, \lambda - \gamma \geq cT
\end{aligned}
\right.
\label{eq:SCPconstraints}
\end{align} 
For the above SCP$_{N}$, we assume that the optimal solution exists and is unique for any possible number $N$ of samples, which is a standard assumption in the literature; see, e.g.,  \cite{calafiore2005uncertain} for more discussion on this assumption.

Note that the decision variables are same in RCP and $\text{SCP}_{N}$. In the rest of paper, given a solution $\hat{d}$, we denote by $\hat{d} \models \mathcal{O}$ if $\hat{d}$ is a feasible solution of $\mathcal{O}$ where $\mathcal{O}$ is a optimization problem.

The following result established in \cite{salamati2021data} shows how to solve Problem~1 based on SCP$_N$. 
\begin{comment}
With assumption above, the authors in \cite{salamati2021data} proved a rior guarantee of safety synthesis, which is stated as below.
\end{comment}

\begin{mythm}[\!\!\cite{salamati2021data}] \label{lemma:prior}
Given   dt-DS $\mathbf{S} = (X,U,f)$ with unknown $f$ and safety property $\varphi=(X_0,X_u,T)$. 
Let $d^{*}_{N}=(K^{*}_{N},\lambda^*,\gamma^*,c^*,q^*,p^*)$ the optimal solution to $\text{SCP}_{N}$ 
and  $C(x)=[F_1(p^{*}_{1},x);\dots; F_m(p^{*}_{m},x)]$ be the associated controller. 
Then we have $\mathbf{S}_{C} \models \varphi$  with a confidence of at least $1-\beta$ if, for some $\epsilon \in [0,1]$, we have  
\begin{align} 
 N \geq   N(\epsilon,\beta) \text{ and }  K^{*}_{N}+L \mathcal{U}^{-1}(\epsilon)  \leq 0, \label{eq:testequation}
\end{align}
where 
\[
N(\epsilon,\beta) := \min\left\{\!\! N \in \mathbb{N} \,\middle\vert\,  \sum_{i=0}^{Q+P+3} \binom{N}{i} (\epsilon)^i(1-\epsilon)^{N-i} \leq \beta \right\}
\]
with $Q$ and $P$ the number of coefficients in the CBF and  the total number of coefficients in control functions, respectively, 
and  $\mathcal{U}(r) : \mathbb{R}_{0}^{+} \rightarrow [0,1]$ is a function related to geometry of $X \times U$ and sampling distribution $\mathbb{P}$.
\end{mythm}

 \begin{remark} \label{remark:uniform set}
 The reader is referred to  \cite{esfahani2015performance} for the general relationship 
 among  function $\mathcal{U}(\cdot)$, distribution $\mathbb{P}(\cdot)$  and space $X\times U$. 
 Particularly, if the the sampling distribution $\mathbb{P}$ is uniform over $X \times U$ and $X \times U$ is $n$-dimensional hyper-rectangular, then function $\mathcal{U}$ is given by \cite{nejati2023formal}
 \[
 \mathcal{U}(r)=\frac{\pi^{\frac{n}{2}}r^{n}}{2^{n}\Gamma(\frac{n}{2}+1)\mathbf{Vol}(X \times U)}
 \]
 where $\Gamma(\frac{n}{2}+1)=(\frac{n}{2}+1)!$ when $n$ is even and $\Gamma(\frac{n}{2}+1)=\frac{n}{2} \times (\frac{n}{2}-1) \times  \cdots \times \frac{1}{2}$ otherwise, and $\mathbf{Vol}(\cdot)$ denotes the volume of a set.
 \end{remark}
 
%  \begin{remark}
%  The dynamic system considered in this work is non-stochastic, the probabilistic confidence bound in Lemma~\ref{lemma:prior} is due to randomly collect data from system.
%  \end{remark}

% When we consider uniform sample distribution $\mathbb{P}$ and rectangular state space, from remark~\ref{remark:uniform set} we know that the required sample number is exponential with respect to dimension of system. However, the approach in Lemma~\ref{lemma:prior} does not use Posteriori information, which we will discuss useful in next subsection, to reduce sample complexity.

% In this work, we assume that the constraints are Lipschitz continous with respect to $x$ and $u$, which is stated formally as follows:
%  \begin{assumption} \label{assumption}
% Function $g_3$ is Lipschitz continuous with respect to $(x,u)$ with Lipschitz constant $\text{L}_3$. Function $g_1$, $g_2$, $g_4$ are also Lipschitz continuous with respect to $x$ with Lipschitz constant $\text{L}_1, \text{L}_2, \text{L}_4$, respectively. We denote by $\text{L}_{x,u}$ the maximum value of all these Lipschitz constants.
% \end{assumption}

%  \begin{remark}
%  The RCP is a robust convex optimization. The constraints and objective function are convex with respect to variables $d$. It is a robust optimization because all constraints needs to hold for all $x \in X$.
%  \end{remark}

\section{Main Results using Posterior Information} \label{section:synthesis}

In the previous section, we reviewed existing methods that provide a sound data-driven solution to Problem 1. However, as noted in Remark \ref{remark:uniform set}, the number of sample data required to achieve a desired confidence bound is generally exponential with respect to the dimension of the system. Therefore, the question naturally arises: How can we improve the sampling efficiency of the synthesis procedure? To address this issue, we present a more efficient method that leverages additional information.

In the context of SCP, there are two additional posteriori information that are closely related to the performance bound of the program: 
\begin{itemize}
    \item 
    one is the \emph{support constraint} whose removal can improve the optimal value of the SCP; 
    \item 
    the other is the \emph{violation frequency} of a new set of validation data.
\end{itemize}
As shown in \cite{shang2021tac}, these two posteriori information can be leveraged together to improve  the sample efficiency in order to achieve a desired performance bound. In this section,  we show how these information can be used in the context of data-driven control synthesis. 

First, we review the definition of  support constraint.
\begin{mydef}[Support Constraint,\cite{campi2008exact}] \label{def:support constraint}
For a scenario convex program $\text{SCP}_{N}$ and $i \in \{ 1,\dots,N \}$, 
constraint $g_{3}(x_i,u_i,d) \leq 0$ is said to be a \emph{support constraint} if the removal of the constraint improves optimal value of $\text{SCP}_{N}$.
\end{mydef}

Intuitively,  the number of support constraints characterizes the complexity of $\text{SCP}_{N}$.  
As shown in \cite{garatti2022risk}, the number of support constraints is upper bounded by the number of decision variables.
Furthermore, if the number of support constraints is much smaller than the number of decision variables, the complexity of $\text{SCP}_{N}$ is much lower than we guess in a prior. It means that we can provide the same guarantee by less samples.

The concept of \emph{violate frequency} arises in the \emph{validation test procedure}. 
Specifically,  suppose that we form an $\text{SCP}_{N}$ from a set $\mathcal{D}$ of $N$ sample data 
and let $d^{*}_{N}$ be the optimal solution to $\text{SCP}_{N}$. 
The validation test requires a new set $\mathcal{D}'$ of $N_0$ independent samples  of state-input pair $\{ (x'_1,u'_k),\dots,(x'_{N_0},u'_{N_0}) \}$.  Then the violation frequency is defined as follows. 

\begin{mydef}[Violation Frequencies,\cite{10.1214/14-EJS909}] \label{def:violationfrequency}
Let $d^{*}_{N}$ be the optimal solution to $\text{SCP}_{N}$ formed  by data set $\mathcal{D}$ with $N$ samples. 
Let $\mathcal{D}'$ be a set of $N_0$ independent new samples. 
Then the  \emph{violation frequency} with respect to $N_0$ and $d^{*}_{N}$ is defined by
\begin{equation} \label{eq:fre}
    R_{N_0}= \sum_{k=1}^{N_0} v(k),
\end{equation}
where $v(k)$ is the   the violation indicator of $d^{*}_{N}$  for the $k$-th sample defined by 
\begin{align}\label{eq:validation}
	v(k) = 
		\left\{
		\begin{array}{ll}
			0 & g_3(x'_k,u'_k,d^{*}_{N}) \leq 0  \\
			1   & \qquad \text{otherwise}
		\end{array}
		\right.   .
\end{align} 
\end{mydef}\medskip

Before we provide our main result, we make the following assumption regarding SCP$_N$. 

\begin{assumption}(Non-degeneracy\cite{campi2018wait}) \label{assumption:nonde}
The solution to $\text{SCP}_{N}$ coincides with probability $1$ with the solution to the program only defined by support constraints.
\end{assumption}

The above assumption is a very mild one for convex programs. It effectively rules out situations where the solution of the program with only support constraints lies on the boundaries of other constraints with a non-zero probability.

% Intuitively, the validation test is also a posteriori information about the $\text{SCP}_{N}$. Specifically, the violation frequency can also provide a guarantee bound between solution of scenario program and robust program \cite{10.1214/14-EJS909}.
% \begin{comment}
% \begin{remark}
% Note that if some sample violates scenario convex program, it does not mean that safety specification is violated. Because the optimal objective value may be negative.
% \end{remark}
% \end{comment}

%establish a probabilistic bride between the solution of $\text{SCP}_{N,\hat{N}}$ and satisfaction of safety specification of dt-SCS $\mathbf{S}$

% Assume that $d^{*}_{N}$ is optimal solution of $\text{SCP}_{N}$. We denote by $N^{*}$ the number of support constraints. Now we state the main result of this work, which connects the safety of a controlled system with the optimal solution $d^{*}_{N}$, violation frequency $R_{N_0}$, and number of supported constraints $N^{*}$ of $\text{SCP}_{N}$.

Now, let $d^{*}_{N}$ be the optimal solution of $\text{SCP}_{N}$ with $N^{*}$ the number of support constraints.  
The following main result of this paper establishes the connection between the safety of a controlled system and the optimal solution of   $\text{SCP}_{N}$,  its number of support constraints and the violation frequency of a new set of data.

\begin{mythm} \label{thm:main}
Given   dt-DS $\mathbf{S} = (X,U,f)$ with unknown $f$ and safety property $\varphi=(X_0,X_u,T)$. 
Let $d^{*}_{N}=(K^{*}_{N},\lambda^*,\gamma^*,c^*,q^*,p^*)$ the optimal solution to $\text{SCP}_{N}$ formed by a set $\mathcal{D}$ of $N$ data and $N^{*}$ be the number of support constraints. 
Let $\mathcal{D}'$ be a collection of $N_0$ new independent data and $R_{N_0}$ be the violation frequency w.r.t.\ $N_0$ and $d^{*}_{N}$. 
Let   $C(x)=[F_1(p^{*}_{1},x);\dots; F_m(p^{*}_{m},x)]$ be  controller associated with $d^{*}_{N}$. 
Then we have $\mathbf{S}_{C} \models \varphi$  with a confidence of at least $1-\beta$ if 
\begin{equation} \label{eq:check_condition}
    K^{*}_{N}+L \mathcal{U}^{-1}(1-\kappa^{*}) \leq 0,
\end{equation}
where $\kappa^{*}$ is the unique solution of 
\begin{equation} \label{eq:support}
     \frac{\beta}{N+1} \sum_{i=N^{*}}^{N} \binom{i}{N^{*}} \kappa^{i-N}- \binom{N}{N^{*}} B_{N_0}(1-\kappa;R_{N_0}) = 0.
\end{equation}
\end{mythm}
\begin{proof}
From RCP we construct a chance constraint program $\text{CCP}_{\epsilon}$ for some $\epsilon$ as below:
\begin{align}
\text{CCP}_{\epsilon}: \left\{
\begin{aligned}
    \min_{d} \quad & K  \\
     \text{s.t.}  \quad &  \mathbb{P}( g_3(x,u,d)  \leq 0 ) \geq 1 - \epsilon, \\
     & \max_{z \in \{1,2,4 \}}g_z(x,d)  \leq 0, \forall x \in X,\\
    & d=(K,\lambda,\gamma,c,q,p), \\
    & K \in \mathbb{R}, c\geq 0, \lambda - \gamma \geq cT
\end{aligned}
\right.
\label{eq:CCPepsilon}
\end{align} 

Using Theorem 4 in~\cite{shang2021tac}, we know that
\begin{equation} \label{eq:CCPfea}
    \mathbb{P}^{N+N_0}(d^{*}_{N} \models \text{CCP}_{\epsilon^*}) \geq 1 - \beta
\end{equation}
where $\epsilon^* = 1- \kappa^{*}$. Then we construct a relax version of RCP, denoted by $\text{RCP}_{h(\epsilon^*)}$ as follows:
\begin{align}
\text{RCP}_{h(\epsilon^*)}: \left\{
\begin{aligned}
    \min_{d} \quad & K  \\
     \text{s.t.}  \quad &   g_3(x,u,d)  \leq h(\epsilon^*), \forall x \in X, \forall u \in U \\
     & \max_{z \in \{1,2,4 \}}g_z(x,d)  \leq 0, \forall x \in X,\\
    & d=(K,\lambda,\gamma,c,q,p), \\
    & K \in \mathbb{R}, c\geq 0, \lambda - \gamma \geq cT
\end{aligned}
\right.
\label{eq:RCPhepsilon}
\end{align} 
where $h(\cdot)$ is a uniform level-set bound defined in Definition 3.1 of~\cite{esfahani2015performance}. From result of Lemma 3.2 in~\cite{esfahani2015performance} and Equation~(\ref{eq:CCPfea}), we know that
\begin{equation} \label{eq:RCPrefea}
    \mathbb{P}^{N+N_0}(d^{*}_{N} \models \text{RCP}_{h(\epsilon^*)}) \geq 1 - \beta.
\end{equation}
We denote by $K^*_{\text{RCP}_{h(\epsilon^*)}}$ the optimal value of objective function of $\text{RCP}_{h(\epsilon^*)}$. Since any feasible solution of $\text{RCP}_{h(\epsilon^*)}$ is larger or equal to $K^*_{\text{RCP}_{h(\epsilon^*)}}$, we have
\begin{equation} \label{eq:optimalfeainRCPh}
    \mathbb{P}^{N+N_0}(  K^*_{\text{RCP}_{h(\epsilon^*)}} \leq  K^*_{N}) \geq 1 - \beta.
\end{equation}
Using Lemma 3.4 in~\cite{esfahani2015performance}, we know that
%\[
%K^*_{N,\hat{N}} \leq K^*_{\text{RCP}_{\hat{N}}} \leq K^*_{\text{RCP}_{h(\epsilon^*)}} + \mathcal{L}_{sp}h(\epsilon^*)
%\]
\[
K^*_{\text{RCP}} \leq K^*_{\text{RCP}_{h(\epsilon^*)}} + \mathcal{L}_{sp}h(\epsilon^*)
\]
where $K^*_{\text{RCP}}$ is optimal of $\text{RCP}$ and $\mathcal{L}_{sp}$ is the Slater constant defined in Assumption 3.3 of~\cite{esfahani2015performance}. As a result, we have
%\[
%    \mathbb{P}^{N+N_0}(K^*_{N,\hat{N}} \leq K^*_{\text{RCP}_{\hat{N}}} \leq K^*_{N,\hat{N}} + \mathcal{L}_{sp}h(\epsilon^*)) \geq 1 - \beta
%\]
\[
    \mathbb{P}^{N+N_0}( K^*_{\text{RCP}} \leq K^*_{N} + \mathcal{L}_{sp}h(\epsilon^*)) \geq 1 - \beta.
\]
Because $\text{RCP}$ is a min-max problem, according to Remark 3.5 in~\cite{esfahani2015performance}, $\mathcal{L}_{sp}$ can be chosen as 1. Moreover, as statement in Remark 3.8 of~\cite{esfahani2015performance}, $h(\epsilon^*)$ can be computed as $L\mathcal{U}^{-1}(\epsilon^*)$ where $L$ is Lipschitz constant of constraints. Thus we have
%\[
%    \mathbb{P}^{N+N_0}(K^*_{N,\hat{N}} \leq K^*_{\text{RCP}_{\hat{N}}} \leq K^*_{N,\hat{N}} + \text{L}_{x,u}\sqrt[n]{\epsilon^*}) \geq 1 - \beta
%\]
\begin{equation} \label{eq:condi}
      \mathbb{P}^{N+N_0}( K^*_{\text{RCP}} \leq K^*_{N} + L\mathcal{U}^{-1}(\epsilon^*)) \geq 1 - \beta .  
\end{equation}
We define
\[
E = \{ \mathcal{D} \in \Delta^{N+N_0} | K^*_{\text{RCP}} \leq K^*_{N} + L\mathcal{U}^{-1}(\epsilon^*) \}
\]
the set of datesets includes in the event of Equation~(\ref{eq:condi}). Let
\[
F = \{ \mathcal{D} \in \Delta^{N+N_0} |K^*_{N} + L\mathcal{U}^{-1}(\epsilon^*) \leq 0 \}.
\]
If $E \cap F \neq \emptyset$, we know that $K^*_{\text{RCP}} \leq 0$. Since $K^*_{N}+L\mathcal{U}^{-1}(\epsilon^*) \leq 0$, we know that the selected date set $\mathcal{D} \in F$. From Equation~(\ref{eq:condi}) we have
\[
\mathbb{P}^{N+N_0}(\mathcal{D} \in E) \geq 1-\beta.
\]
Therefore, $K^*_{\text{RCP}}\leq 0$, i.e., $\mathbf{S}_{C} \models \varphi$, is true with confidence of at least $1-\beta$. This completes the proof.
\end{proof}

\begin{remark}
Before we proceed further, let us discuss some computational considerations regarding the derived performance bound. First, we can obtain an upper bound of the Lipschitz constant $L$ in Theorem~\ref{thm:main} by using the result in Lemma 2 of~\cite{salamati2021data}. Second, in cases where the number of constraints is large, it may be challenging to accurately count the number of support constraints. However, for convex optimization problems, the support constraint is also an active constraint~\cite{boyd2004convex}. Therefore, we can use the number of active constraints in $\text{SCP}_{N}$ as an upper bound for the number of support constraints. Finally, we note that the solution of Equation~(\ref{eq:support}) may not have an analytic expression. Nevertheless, we can use bisection to numerically search for the solution using the procedure in Algorithm 1 of~\cite{shang2021tac}.
\end{remark}

\begin{algorithm}[!htp]
\caption{Data-driven system Safe  Control Synthesis for  Unknown dt-DS}\label{alg:overall} 
\KwIn{$\mathbf{S}=(X,U,f)$, $\varphi=(X_0,X_u,T)$, $\beta \in [0,1]$, $L \in \mathbb{R}$, degree $k,k' \in \mathbb{N}$ in (\ref{eq:CBFcoe}) and (\ref{eq:inputcoe}). }

Select a probability distribution $\mathbb{P}$ over $X \times U$\\
Choose number of samples $N$ and $N_0$   \\
Collect $N$ samples $\mathcal{D}_1=\{ (x_i,u_i,x'_{i}) \in X \times U \times X \mid x'_{i}=f(x_i,u_i)  \}$  \\
Collect $N_0$ additional samples $\mathcal{D}_2=\{ (x_i,u_i,x'_{i}) \in X \times U \times X \mid x'_{i}=f(x_i,u_i)  \}$  \\
Solve $\text{SCP}_{N}$ by $\mathcal{D}_1$ and obtain $d^{*}=(K^{*},\lambda^*,c^*,q^*,p^*)$ \\
Compute the  number of support constraints $N^*$ \\
Compute violation frequency $R_{N_0}$ by $\mathcal{D}_2$ \\
Compute $\kappa^{*}$ according to Equation~(\ref{eq:support}) \\

\KwOut{Controller $C$ defined by Equation~(\ref{eq:controller}) has guarantee that $\mathbf{S}_{C} \models \varphi$ with confidence $1-\beta$  if $K^{*}_{N}+L\mathcal{U}^{-1}(1-\kappa^{*}) \leq 0$}
\end{algorithm}

Now we discuss how to properly select sample numbers $N$ and $N_0$ to achieve confidence $\beta$. For each $N$, we can solve $\text{SCP}_{N}$ repetitively to estimate optimal objective value $\widehat{K}^*_{N}$ and support constraints number $\widehat{N}^*$ in expectation. According to analysis in~\cite{shang2021tac}, the expectation of $R_{N_0}/N_0$ is lower than $N^*/N$ with high confidence. Therefore, we adopt $\widehat{R}_{N_0}=N_0 \times \widehat{N}^*/N$ as estimated violation frequency.

Since $\mathcal{U}(\cdot)$ is an increasing function, to guarantee Equation~(\ref{eq:check_condition}), we have
\begin{equation} \label{eq:mid_condtion}
    \kappa^* \geq 1-\mathcal{U}\left(-\frac{\widehat{K}^*_N}{L}\right).
\end{equation} 
We denote by $g^{\beta}_{N,N_0}(\kappa,N^*,R_{N_0})$ the LHS of Equation~(\ref{eq:support}). From~\cite{shang2021tac} we know that function $g^{\beta}_{N,N_0}(\cdot)$ is decreasing w.r.t. $\kappa$. Therefore, we can pick $N$ and $N_0$ such that 
\begin{equation} \label{eq:pick_condition}
    g^{\beta}_{N,N_0}\left(1-\mathcal{U}\left(-\frac{\widehat{K}^*_N}{L}\right),\widehat{N}^*,\widehat{R}_{N_0}\right) \geq 0.
\end{equation}

In summary,  given a desired confidence bound $\beta$, we can determine the $N$ and $N_0$ by following steps:
\begin{enumerate}
    \item Pick $N$ and $N_0$ arbitrary;
    \item Calculate $\widehat{K}^*_N$, $\widehat{N}^*$, $\widehat{R}_{N_0}$ as discussion above;
    \item If Equation~(\ref{eq:pick_condition}) holds, then choose the current $N$ and $N_0$; otherwise increase $N$ and $N_0$ and return to step 2). 
\end{enumerate} 

It is important to note that while following the steps outlined above to select values for $N$ and $N_0$, it is possible that condition~(\ref{eq:check_condition}) may not be satisfied. In such cases, we must independently sample a new set of $N+N_0$ data and solve a new instance of $\text{SCP}_{N}$. In order to reduce the time required, we can over-approximate $K^*_N$ and $N^*$ to obtain a more conservative bound. The overall steps involved in this process are summarized in Algorithm~\ref{alg:overall}.

\section{Case Study of Room Temperature Control} \label{section:casestudy}
To illustrate the efficiency of the proposed approach, we adopt the room temperature control problem from \cite{salamati2021data}. 
Specifically, we control room with a heater whose dynamic function is given by  
\[
\mathbf{S} : x(t+1)=x(t)+\tau_s ( \alpha_e (T_e-x(t)) + \alpha_h (T_h-x(t)) u(t) ), 
\]
where $T_e = 15$, $T_h = 45$, $\alpha_e = 8 \times 10^{-3}$, $\alpha_h = 3.6 \times 10^{-3}$ and $\tau_s = 5$. 
We define $X_0 = [24,25]$, $X_u =[22.5,23] \cup [26,26.5]$, $X=[22.5,26.5]$, $U=[0,1]$ and $T=5$. 
We assume that the dynamic of system is unknown and the objective is to synthesize a controller under which the room temporal $x(t)$ in  a comfortable region between $23^{\circ}$ and $26^{\circ}$  within time horizon $T=5$  with confidence of $95\%$.

For both CBF and controller function, we consider candidate polynomial functions with degree $k=k_1=4$. 
Then the   CBF and controller function are in the form of 
\[
\mathcal{B} (q,x) = \mathbf{x} \mathbb{Q} \mathbf{x}^{\top} \text{ and } C(p,\mathbf{T}) = \mathbf{x} \mathbb{P} \mathbf{x}^{\top},
\]
where $\mathbf{x}=[1,x,x^2]$ is a row vector  and 
\begin{equation}
    \mathbb{Q}= 
   \begin{bmatrix}
   q_0 & \frac{q_1}{2} & \frac{q_2}{3} \\
   \frac{q_1}{2} & \frac{q_2}{3} & \frac{q_3}{2} \\
    \frac{q_2}{3} & \frac{q_3}{2} & q_4 
    \end{bmatrix},
    \mathbb{P} = 
    \begin{bmatrix}
   p_0 & \frac{p_1}{2} & \frac{p_2}{3} \\
    \frac{p_1}{2} & \frac{p_2}{3} & \frac{p_3}{2} \\
   \frac{p_2}{3} & \frac{p_3}{2} & p_4 
    \end{bmatrix} 
\end{equation}
are two coefficient matrices. 
By enforcing $\Vert  \mathbb{Q} \Vert \leq 0.1$ and $\Vert  \mathbb{P} \Vert \leq 0.05$, the Lipschitz constant $L$ can be upper bounded by $11.63$. We choose uniform distribution to sample state-input pairs. Since the state-input space is a $2$-dimensional rectangular,  function $\mathcal{U}(r)$ is computed by
$\mathcal{U}(r)=\frac{\pi}{16}r^2$.

\textbf{Results by Existing Method: }
First, we use the results of \cite{salamati2021data} as stated in Theorem~\ref{lemma:prior} to solve the data-driven control synthesis problem.  Let $\beta=0.05$. 
By  choosing $\epsilon=7.492\times 10^{-6}$, we have $L\mathcal{U}^{-1}(\epsilon)=0.07$. 
Therefore, the minimum number of samples needed for the scenario convex program to ensure the confidence bound is $N=2733296$. 
We then solve the $\text{SCP}_{N}$ with acquired samples and obtain the optimal objective value $K^{*}_{N}=-0.1486$. 
Since $K^{*}_{N}+L\mathcal{U}^{-1}(\epsilon)=-0.0786 \leq 0$, we know that $\mathbf{S}_{C} \models \varphi$ is ensured with a confidence of at least $1-\beta = 95\%$. 

\textbf{Results by Our Method: }
Now we consider the posteriori method proposed in this paper. 
We also select $\beta=0.05$. Note that there is no need to fix $\epsilon$ in a priori. 
Here, we choose $N=140000$ to form the SCP and choose $N_0=70000$ for the validation test.
Then we solve $\text{SCP}_{N}$ and obtain $K^*_{N}=-0.149$, $\lambda^*=-68.14$, $\gamma^*=-69.64$ and $c^*=0.2998$. 
We use number of active constraints, which is $1$, to upper bound the number of support constraints. 
In the validation test, the violation frequency is $0$, which essentially means that the solution to the SCP is already good enough to deserve higher confidence bound. 
The solution of Equation~(\ref{eq:support}) is $\kappa^{*}=0.9999723$. 
Since $K^*_{N}+L\mathcal{U}^{-1}(1-\kappa^{*})=-0.011 \leq 0$,  we know that $\mathbf{S}_{C} \models \varphi$ is ensured with a confidence of at least $1-\beta = 95\%$.  
The CBF computed from $\text{SCP}_{N}$ is 
\begin{align}
 \mathcal{B} (x)=&1.948\times 10^{-3}x+0.2395x^2 \nonumber\\
 &-3.841\times 10^{-2}x^3+9.740\times 10^{-4} x^4 \nonumber
\end{align} 
and the obtained controller is
\begin{align}
C(x)=&1.208\times 10^{-5} +9.768 \times 10^{-2} x -3.438 \times 10^{-3} x^2 \nonumber\\
 & +2.418\times 10^{-5} x^3+4.594 \times 10^{-7} x^4. \nonumber
\end{align}  
The constructed $\mathcal{B} (x)$ is shown in Figure~\ref{fig:barrier}. From Figure~\ref{fig:barrier} we know that conditions~(\ref{barrier:condition 1}) and (\ref{barrier:condition 2}) are satisfied.
Since we know the underlying dynamic of system $\mathbf{S}$, we also draw constraint $g_3(\cdot)$ in Figure~\ref{fig:con3}, which shows that condition~(\ref{barrier:condition 3}) is also satisfied. Therefore, $\mathcal{B} (x)$ is indeed a CBF and $\mathbf{S}_{C} \models \varphi$, i.e., controlled system is safe.

\begin{figure} 
  \centering
  \includegraphics[width=6cm]{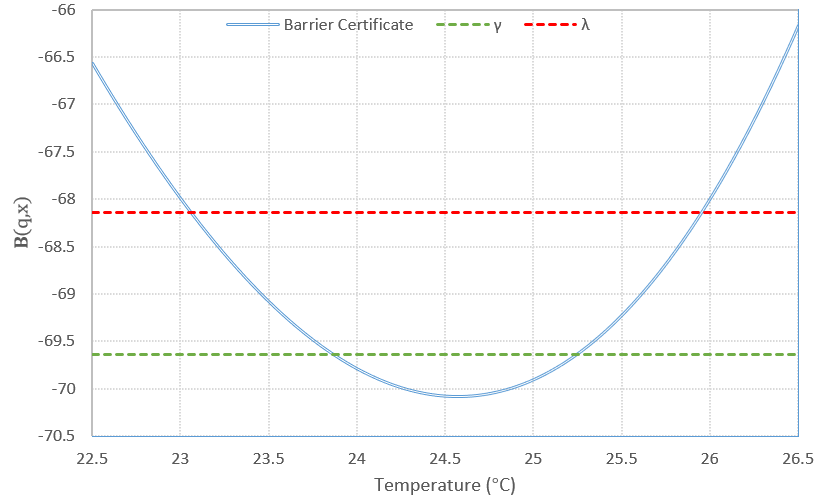}
  \caption{Computed $\mathcal{B} (x)$ of $\text{SCP}_{N}$. The green and red dashed line represents solution $\gamma^*$ and $\lambda^*$, respectively.}
 \label{fig:barrier}
\end{figure}

\begin{figure} 
  \centering
  \includegraphics[width=6cm]{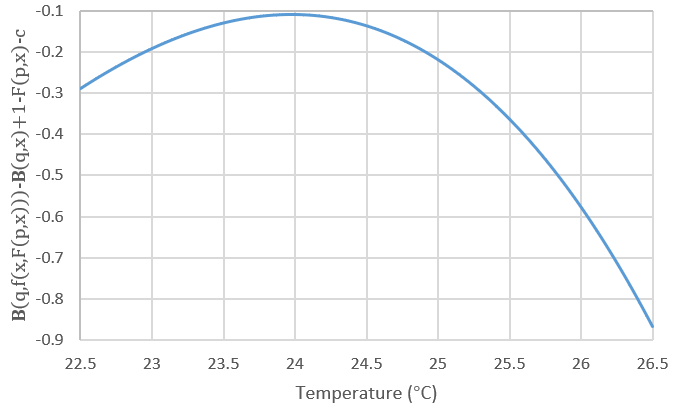}
  \caption{Satisfaction of condition~\ref{barrier:condition 3} of computed $\mathcal{B} (x)$ and $C(x)$.}
 \label{fig:con3}
\end{figure}

In the above example, we get zero violation frequency for the experiment. To further show the average performance, we run  Algorithm~\ref{alg:overall} for $100$ times with $N=140000$ and $N_0=70000$. The number of active constraint is always $1$. We record  $R_{N_0}$ in the $100$ runs in Table~\ref{table:record}. As Table~\ref{table:record}, $R_{N_0}=0$ has the highest occurrence frequency. It has been shown in~\cite{10.1007/s10107-019-01446-4} that we have  high confidence that $R_{N_0}/N_0$ can not be much higher than $N^*/N$, where $N^*$ is number of support constraints. Since $N^* \leq 1$ mostly in $\text{SCP}_{N}$, we know that $R_{N_0}$ cannot much higher than $0.5$ with high confidence. Therefore, the outcome of Table~\ref{table:record} is consistent with the result in~\cite{10.1007/s10107-019-01446-4}. Moreover, the expected number of samples $N_{e}$ needed for our method  is $N_{e}=(N+N_0)/0.42=500000$.

\begin{table}[!htp]
  \begin{center}
    \caption{Record of $R_{N_0}$ for 100 runs of Algorithm~\ref{alg:overall}}
    \label{table:record}
    \begin{tabular}{|c|c|c|c|c|c|} % <-- Alignments: 1st column left, 2nd middle and 3rd right, with vertical lines in between
    \hline
     $R_{N_0}$  & 0 & 1 & 2 & 3  & 6\\
      \hline
      frequency  & 42 & 34 & 17 & 5 & 2 \\
      \hline
    \end{tabular}
  \end{center}
\end{table}

According to the above experiments, 
our approach uses a significantly smaller sample size compared to the method proposed in~\cite{salamati2021data}. This is mainly because the number of support constraints is much lower than the number of decision variables. This observation suggests that the complexity of the SCP is considerably lower than we initially assumed. Moreover, the frequency of violations is low in most cases, indicating that our solution provides the desired level of confidence. Collectively, these findings highlight the effectiveness of our approach in addressing the SCP with a reduced sample size while ensuring the same level of confidence.

\section{Conclusion}\label{sec:con}

In this work, we presented a new approach for synthesizing safety controllers for unknown dynamic systems using data. Our method involves solving a scenario convex program formed by the data, followed by a validity test to improve confidence. To achieve this, we derived a new overall performance bound that combines the information from the original sample data, support constraints, and violation frequency. Our experiments demonstrated that our approach is more sample-efficient than existing methods. In this work, we only consider deterministic dynamic systems. In the future, we plan to extend our results to the stochastic case.  

\bibliographystyle{ieeetr}
\bibliography{STL}   

\end{document}